\documentclass[Afour,sageh,times]{sagej}

\usepackage{mathrsfs}  

\usepackage[colorlinks,bookmarksopen,bookmarksnumbered,citecolor=red,urlcolor=red]{hyperref}

\newtheorem{theorem}{Theorem}{}
\newtheorem{remark}{Remark}{}

\newcommand\BibTeX{{\rmfamily B\kern-.05em \textsc{i\kern-.025em b}\kern-.08em
T\kern-.1667em\lower.7ex\hbox{E}\kern-.125emX}}

\def\volumeyear{2020}

\begin{document}

\runninghead{Tasoujian et al.}

\title{Robust linear parameter-varying output-feedback  control of permanent magnet synchronous motors}

\author{Shahin Tasoujian, Jaecheol Lee, Karolos Grigoriadis, and Matthew Franchek}

\affiliation{This paper is a preprint of a paper submitted to Transactions of the Institute of Measurement and Control. \\ Department of Mechanical Engineering at the University of Houston, Houston, TX 77004 U.S.A.}
% \affilnum{2}SAGE Publications Ltd, UK}

\corrauth{Shahin Tasoujian
Department of Mechanical Engineering, University of Houston, Houston, TX 77004 U.S.A.}

\email{stasoujian@uh.edu}

\begin{abstract}
This paper investigates the design of a robust output-feedback linear parameter-varying (LPV) gain-scheduled controller for the speed regulation of a surface permanent magnet synchronous motor (SPMSM).  Motor dynamics is defined in the $\alpha - \beta$ stationary reference frame and a parameter-varying model formulation is provided to describe the SPMSM nonlinear dynamics. In this context, a robust gain-scheduled LPV output-feedback dynamic controller is designed to satisfy the asymptotic stability of the closed-loop system and meet desired performance requirements, as well as, guarantee robustness against system parameter perturbations and torque load disturbances. The real-time impact of temperature variation on the winding resistance and magnet flux during motor operations is considered in the LPV modelling and the subsequent control design to address demagnetization effects in the motor response. The controller synthesis conditions are formulated in a convex linear matrix inequality (LMI) optimization framework.  Finally, the validity of the proposed control strategy is assessed in simulation studies, and the results are compared to the results of the conventional field-oriented control (FOC) method. The closed-loop simulation studies demonstrate that the proposed LPV controller provides improved transient response with respect to settling time, overshoot, and disturbance rejection in tracking the velocity profile under the influence of parameter and temperature variations and load disturbances.
\end{abstract}


\maketitle

\section{Introduction}
\label{introduction}
Permanent magnet synchronous motors (PMSMs) are prevalent in industry and in various electromechanical applications, such as, electrical appliances, robotic systems and electric vehicles, due to their compact structure, high torque density, high power density and high-efficiency \cite{boldea1992vector, zhong1997analysis, yanliang2001development}. However, because of the inherent nonlinear dynamics, strong coupling effects and significant system parameter variability \cite{pillay1988modeling, cai2017optimal}, the precise speed and position control of a PMSM is a challenging task. Traditionally, the field-oriented control (FOC) method has been employed as a vector control of both magnitude and angle of the flux enabling independent control of torque and speed. Consequently, fast and high precision motor control can be achieved. For this reason, the motor drives implemented with the FOC method are typically comprised of two loops in a cascade manner in the $d-q$ rotating reference frame \cite{zhu2019performance}. The current control loop is the inner loop for the stator current to follow its reference value while the speed control loop is the outer loop taking into account speed error signals and providing reference signals to the inner loop \cite{giri2013ac}.

In the FOC method, a proportional-integral (PI) controller is typically implemented for both current and speed control due to its simple design structure. The PI controller gains are typically determined through nominal motor parameters to satisfy motor performance specifications \cite{kim2016self}. However, PI control is not suitable for applications where high performance and high precision is required. When motor parameter variations and disturbances are present, robustness and stability issues inevitably arise. As an additional challenge, varying motor temperature has shown to have a significant impact on PMSM speed, current and torque resulting from the reversible demagnitization of the permanent magnet (NdFeB or SmCo) and the temperature-dependence of the stator winding resistance. Hence, traditional PI controllers typically fail to maintain the desired closed-loop motor response in high performance applications. Various robust and nonlinear control methods have been adopted to address the parameter variability, as well as, to cope with the nonlinearity in the PMSM model \cite{zhao2019robust}. The sliding mode control (SMC) method has been proposed to assure fast response and robustness in the presence of nonlinearity in the model. However, the SMC method inherently causes a chattering problem due to the signum function, which leads to deteriorating performance at steady-state \cite{baik1998robust, kim2010high, zhang2012nonlinear}. Disturbance observers (DOBs) have been proposed to estimate the disturbance for its compensation. \cite{zhao2015adaptive} studied the case with unknown load torque and model parameters and proposed an adaptive observer-based control method for the speed tracking in the PMSMs. Although DOBs can help improve the capability of a motor to reject disturbances, a disadvantage lies in the fact that the methodology is required to have full knowledge of the PMSM parameters to ensure the stability of the DOBs \cite{solsona2000nonlinear, chang2010robust}. Additionally, the fuzzy logic control method has been proposed for the control of PMSMs and has shown an improved performance regarding robustness to disturbance rejection. However, shortcomings reside in the fact that membership functions rely solely on the designer's experience and it demands heavy computations \cite{yu2007fuzzy, chaoui2011adaptive}. 

Recently, the use of linear parameter-varying (LPV) gain-scheduling control techniques for the PMSM control problem has drawn the attention of researchers due to the controller's scheduling nature providing the ability to handle system parameter variations and nonlinearities in a systematic framework. In this regard, a static fixed-gain state-feedback LPV controller with an estimator has been proposed for PMSM control \cite{lee2017lpv}, where the estimator is utilized to provide the state-feedback controller with the full-state information needed to generate the control input. The authors in \cite{lee2017lpv} used the polytopic LPV description resulting in a relatively conservative control design, especially for the case of slow parameter variations. The gain-scheduling control technique is an extension of the linear control design to handle nonlinear and time variations, where a scheduling parameter vector captures the information about the nonlinearities or time-varying behavior of the system. The LPV gain-scheduling control methodology was first introduced in \cite{shamma1991guaranteed} to overcome the shortcoming of conventional gain-scheduling control techniques, namely, lack of closed-loop stability and performance guarantees.  Unlike conventional gain-scheduling design methods which are based on interpolation between several independently designed LTI controllers for different fixed operating points, LPV gain-scheduling control design provides a direct, efficient, systematic and global control approach, which also guarantees closed-loop stability and performance. Stability analysis and control synthesis of LPV systems have been addressed extensively in the control literature in the past decade  \cite{apkarian1998advanced, wu2001lpv, tasoujian2019delay, tasoujian2019robust, tasoujian2020robust}. 

In the present paper, first, the $\alpha-\beta$ stationary reference framework is considered for the surface permanent magnet synchronous motors (SPMSMs) modeling.  The SPMSM model is assumed to be subject to varying parameters and torque load disturbances that impair the response of the closed-loop system to track a reference speed profile. Subsequently, we develop a LPV representation to describe the SPMSM dynamics. Resistance and magnetic fluxes in SPMSMs vary with temperature. To this end, temperature variation is taken into consideration in the LPV modeling as an LPV scheduling parameter. The presented formulation allows a systematic control design seeking to handle the temperature-dependent parameter variations and the model uncertainties in SPMSMs. To minimize the conservatism of the control design in meeting performance specifications, a parameter-dependent Lyapunov function approach is utilized to design an LPV gain-scheduled dynamic output-feedback controller to track the commanded reference speed profile and minimize the effect of disturbances and parameter variations over the entire operating envelope of the motor. The proposed dynamic LPV control design method guarantees asymptotic stability and robustness against disturbances and uncertainties in terms of the closed-loop system's induced $\mathcal{L}_2$-norm performance index. A linear matrix inequality (LMI) framework is adopted to formulate the proposed $\mathcal{H}_{\infty}$ control synthesis problem in a convex, computationally tractable setting, which can be solved efficiently using numerical optimization algorithms. Finally, the performance of the proposed method is evaluated and validated in a computer simulation environment and compared to the conventional FOC method with a fixed-gain PI controller.

The notation to be used in the paper is standard and as follows: $\mathbb{R}$ denotes the set of real numbers, and $\mathbb{R}^n$ and $\mathbb{R}^{k \times m}$ are used to denote the set of real vectors of dimension $n$ and the set of real $k \times m$ matrices, respectively. $M \succ \mathbf{0}$ shows the positive definiteness of the matrix $M$  and the transpose of a real matrix $M$ is shown as $M^{\text{T}}$ . Also, $\mathbb{S}^{n}$ denotes the set of real symmetric $n \times n$ matrix. In a symmetric matrix, terms denoted by asterisk, $\star$, will be induced by symmetry as shown below:
\begin{equation*}
    \left[\!\!\begin{array}{cc} 
         S\! +\! W +\! J +\! (\star) &  \!\!\!\star \\
         Q &\!\!\! R 
    \end{array}\!\!\right] :=\! \left[\!\!\begin{array}{cc}
         S \!+ W \!+ W^\text{T} \!+ J \!+ J^\text{T} & \!\! Q^\text{T} \\
         Q & \!\! R
    \end{array}\!\!\right]
\end{equation*}
\noindent where $S$ is symmetric. $\mathbf{He} [\mathbf{M}]$ is Hermitian operator defined as $\mathbf{He} [\mathbf{M}] \triangleq \mathbf{M} + \mathbf{M}^{\text{T}}$ and $\mathcal{C} (J,\: K)$ stands for the set of continuous functions mapping a set $J$ to a set $K$. 

The outline of the paper is as follows. Section \textcolor{red}{II} presents the mathematical modeling for the SPMSMs and the proposed LPV model formulation. In Section \textcolor{red}{III}, the output-feedback LPV gain-scheduling control technique is described considering scheduling parameters that capture the nonlinearity and temperature-dependent variability of the SPMSM model. Section \textcolor{red}{IV} outlines the closed-loop results and describes the performance evaluation of the proposed LPV controller in a computer simulation environment. Finally, Section \textcolor{red}{V} concludes the paper.

\section{SPMSM modeling}
\label{sec:modeling}
\subsection{SPMSM dynamics}
We consider a three-phase synchronous motor with permanent magnets where the magnetic coupling between the phases and the inductance variation due to magnetic saturation are assumed to be negligible. Additionally, the magnetic flux ganerated by the excitation is assumed to have an ideal sinusoidal density distribution. Consequently, the simplified dynamic model for SPMSMs can be expressed in the $\alpha - \beta$ stationary reference  frame as follows \cite {Hwang2014H2CB}:
\begin{align}
      \dot{\theta} & = \omega, \nonumber\\[0.0cm]
    \dot{\omega} & =   \dfrac{1}{J_{m}}\left(-B\omega-K_{t}\textrm{sin}(p\theta)i_{\alpha}+K_{t}\textrm{cos}(p\theta)i_{\beta}-\tau_{L}\right), \nonumber\\[0.0cm]
    \dot{i}_{\alpha}&=\dfrac{1}{L_{s}} \left( -R_{s} i_{\alpha} + p\lambda_{pm} \omega \textrm{sin}(p\theta) + v_{\alpha} \right), \nonumber\\[0.0cm]
    \dot{i}_{\beta}&=\dfrac{1}{L_{s}} \left(-R_{s}i_{\beta}-p\lambda_{pm} \omega\textrm{cos}(p\theta)+v_{\beta}\right),
    \label{PMSM_model}  
\end{align}
\noindent where $\theta$ stands for the mechanical rotor angular position [rad], $\omega$ is the mechanical rotor speed [rad/sec], $v_{\alpha}$, $v_{\beta}$ and $i_{\alpha}$, $i_{\beta}$ are the voltages [V] and currents [A] in the $\alpha - \beta$ stationary reference  frame. In this model, $L_{s}$ denotes the stator inductance [H], $R_{s}$ is the stator resistance [$\Omega$], $\lambda_{pm}$ is the magnetic flux of the motor [Wb], $J_m$ is the moment of inertia $[kg.m^2]$, $p$ denotes the number of magnet pole pairs, $B$ is the viscous friction coefficient [$\textrm{N}\cdot\textrm{m}\cdot\textrm{sec/rad}$],  $\tau_{L}$ is the load torque [$\textrm{N}\cdot\textrm{m}$], and  $K_{t}=\dfrac{3}{2} p \lambda_{pm}$ is the torque constant $[V\cdot\textrm{rad/sec}]$.   
To assess the closed-loop SPMSM performance, the following tracking errors are defined for the quantities of interest: 
\begin{align}
           e_{w}&=\omega^* - \omega, \nonumber\\
       \displaystyle e_{z}&=\int_{0}^{t} e_{\omega} dx,  \nonumber\\
         e_{\alpha}&=i_{\alpha}^*-i_{\alpha}, \nonumber\\
         e_{\beta}&=i_{\beta}^*-i_{\beta},
    \label{tracking_error}
\end{align}
\noindent where $\omega^*$ is the desired motor speed, and $i_{\alpha}^*$ and $i_{\beta}^*$ are the desired currents in the stationary reference $(
\alpha - \beta)$ frame, respectively. Additionally, $e_{z}$ represents the integral of speed error and $e_{\alpha}$ and $e_{\beta}$ are the current errors in the $\alpha - \beta$ stationary reference frame,  respectively. 
% \begin{align}
%             \dot{e}_{z} &= e_{\omega}, \nonumber\\[0.0cm]
%         \dot{e}_{\omega} &=\dfrac{1}{J_{m}} \left( J_{m}\dot{\omega}^*+B\omega+K_{t}\textrm{sin}(p\theta)i_{\alpha}-K_{t}\textrm{cos}(p\theta)i_{\beta}+\tau_{L} \right), \nonumber\\[0.0cm]
%          \dot{e}_{\alpha} &=\dfrac{1}{L_{s}} \left( L_{s}\dot{i}_{\alpha}^*+R_{s}i_{\alpha}-p\lambda_{pm}\omega\textrm{sin}(p\theta)- v_{\alpha} \right), \nonumber\\[0.0cm]
%          \dot{e}_{\beta} &=\dfrac{1}{L_{s}} \left( L_{s}\dot{i}_{\beta}^*+R_{s}i_{\alpha}+ p\lambda_{pm}\omega\textrm{cos}(p\theta)-v_{\beta} \right).
%     \label{error_dynamics}
% \end{align}
The desired torque, $\tau^*$, the desired currents, $i_{\alpha}^*$ and $i_{\beta}^*$, and the voltage inputs to the motor, $v_{\alpha}$ and $v_{\beta}$ are defined as follows
\begin{align}
        \tau^*&=J_{m}\dot{\omega}^*+B\omega^*, \nonumber\\[0.0cm]
         i_{\alpha}^*&=-\dfrac{\tau^*\textrm{sin}(p\theta)}{K_{t}}, \nonumber\\[0.0cm]
         i_{\beta}^*&=\dfrac{\tau^*\textrm{cos}(p\theta)}{K_{t}}, \nonumber\\[0.0cm]
         v_{\alpha}&=L_{s}\dot{i}_{\alpha}^*+R_{s}i_{\alpha}^*-p\lambda_{pm}\omega^*\textrm{sin}(p\theta)-u_{\alpha}, \nonumber\\[0.0cm]
         v_{\beta}&=L_{s}\dot{i}_{\beta}^*+R_{s}i_{\beta}^*+p\lambda_{pm}\omega^*\textrm{cos}(p\theta)-u_{\beta}, 
    \label{5}
\end{align}
    
\noindent where $u_{\alpha}$ and $u_{\beta}$ are the control inputs. Hence, the error dynamics can be obtained by combining (\ref{PMSM_model}), (\ref{tracking_error}), and (\ref{5}) as follows
% By substituting the defined values in (\ref{5}) into the error dynamics (\ref{error_dynamics}), the error dynamics can be rewritten as follows
\begin{align}
            \dot{e}_{z}&=e_{w}, \nonumber\\[0.0cm]
        \dot{e}_{\omega}&=\dfrac{1}{J_{m}}\left(-Be_{\omega}-K_{t}\textrm{sin}(p\theta)e_{\alpha}+K_{t}\textrm{cos}(p\theta)e_{\beta}+\tau_{L} \right), \nonumber\\[0.0cm]
        \dot{e}_{\alpha}&=\dfrac{1}{L_{s}} \left(-R_{s}e_{\alpha}+p\lambda_{pm}\textrm{sin}(p\theta)e_{\omega}+u_{\alpha} \right), \nonumber\\[0.0cm]
         \dot{e}_{\beta}&=\dfrac{1}{L_{s}}\left(-R_{s}e_{\beta}-p\lambda_{pm}\textrm{cos}(p\theta)e_{\omega}+u_{\beta}\right). 
    \label{final_error_dynamics}
\end{align}
Subsequently, an LPV representation for the introduced SPMSM error dynamics is developed to enable LPV control design: 
\subsection{LPV model formulation}
In order to be able to implement the proposed LPV control methodology to the SPMSM dynamics case study, we first rewrite the described system (\ref{final_error_dynamics}) as a proper LPV model. LPV systems correspond to a class of linear systems, whose dynamics depend on time-varying parameters, known as the scheduling parameters. Therefore, considering (\ref{final_error_dynamics}), the first two LPV scheduling parameters are defined as follows
\begin{align}
             \rho_1(\theta(t)) &= p\lambda_{pm}\textrm{sin}(p\theta(t)), \nonumber\\ 
          \rho_2(\theta(t)) &= p\lambda_{pm}\textrm{cos}(p\theta(t)).
    \label{eq:sch_par2}
\end{align}
\noindent Since the scheduling parameters in (\ref{eq:sch_par2}) are trigonometric functions, they can be bounded as follows
\begin{equation}
    \begin{array}{cc}
        -p\lambda_{pm}\leq\rho_{1}(\theta(t))\: \textrm{and} \: \rho_{2}(\theta(t)) \leq p\lambda_{pm}.   \end{array}
\end{equation}
It is known that temperature variation has a significant effect on SPMSM performance. Consequently, we define temperature as the third scheduling parameter:
\begin{equation}
    \begin{array}{cc}
        \underline{T} \leq \rho_{3}(t) = T(t) \leq \overline{T},  
    \end{array}
    \label{8}
\end{equation}
\noindent where $\underline{T}$, and $\overline{T}$ are the minimum and maximum motor operating temperatures in $^\circ C$, respectively. Resistance and magnetic fluxes vary considerably throughout the motor operation as a function of temperature. Embedded insulate temperature sensors or estimation algorithms can be used to provide instantaneous measurements or estimates of stator winding temperature \cite{jun2018temperature}. 
The following relations can be used to obtain empirical expression for these motor parameter variations as functions of temperature
 \begin{align}
             R_{s}(\rho_{3}(t))&=R_{s0} \left( \dfrac{235+\rho_{3}(t)}{310}    \right), \nonumber\\[0.2cm]
        \lambda_{pm}(\rho_{3}(t))&=\lambda_{pm0} \left( 1+\dfrac{\alpha(\rho_{3}(t)-30)}{100}   \right),
     \label{9}
 \end{align}
 \noindent where $R_{s0}$ is the resistance value of the winding at 75$^\circ C$, $\lambda_{pm0}$ is the flux of the magnet at 30$^\circ C$, and $\alpha$ is the temperature coefficient of the magnet in $\%/^\circ C$ \cite{sul2011control}. 
After defining the scheduling parameters, the scheduling parameter vector is represented as, $\boldsymbol{\rho}(t)=[\begin{array}{ccc}
\rho_1(t) & \rho_2(t) & \rho_3(t)\end{array} ]^{\text{T}}$. Subsequently, the LPV representation of the SPMSM dynamics takes the following matrix-vector form
 \begin{align}
               \dot{\mathbf{e}}(t) &= \mathbf{A}(\boldsymbol{\rho}(t))\mathbf{e}(t) +\mathbf{B}_{1}\tau_{L}(t) + \mathbf{B}_{2}\mathbf{u}(t), \nonumber\\
          \mathbf{y}(t) &= \mathbf{C}\: \mathbf{e}(t),
     \label{eq:lpverrordynamics}
 \end{align}
\noindent where the augmented state vector is defined as $\mathbf{e}(t) = \left[\,e_{z}(t)\quad e_{\omega}(t)\quad e_{\alpha}(t)\quad e_{\beta}(t) \right]^{\text{T}}$, the control input is $\mathbf{u}(t)=\left[\, u_{\alpha}(t)\quad u_{\beta}(t) \, \right]^{\text{T}}$, \: $\mathbf{y}(t)$ is the measured signal vector, and the state-space matrices of the LPV system (\ref{eq:lpverrordynamics}) are as follows

\begin{align}
\mathbf{A}(\boldsymbol{\rho}(t)) = & \begin{bmatrix}
       \quad 0 & 1 & 0 & 0 \\[0.1cm]
       \quad 0 & -\dfrac{B}{J_{m}} & -\dfrac{3}{2}\dfrac{\rho_{1}(t)}{J_{m}} & \dfrac{3}{2}\dfrac{\rho_{2}(t)}{J_{m}} \\[0.3cm]
       \quad 0& \dfrac{\rho_{1}(t)}{L_{s}}& -\dfrac{R_{s}(\rho_{3}(t))}{L_{s}}& 0 \\[0.3cm]
       \quad 0& -\dfrac{\rho_{2}(t)}{L_{s}} & 0  & -\dfrac{R_{s}(\rho_{3}(t))}{L_{s}}
     \end{bmatrix},\nonumber\\
\mathbf{B}_{1}\! =\! \begin{bmatrix}
\, 0 \, \\
\, 1 \, \\
\, 0 \,\\
\, 0 \,
\end{bmatrix}\!&, 
\mathbf{B}_{2} = \begin{bmatrix}
\, 0& 0 \, \\
\, 0 & 0 \, \\
\, \dfrac{1}{L_{s}} & 0 \, \\
\, 0 & \dfrac{1}{L_{s}} \,
\end{bmatrix}\!,
\, \mathbf{C} = \begin{bmatrix}
\, 1&0&0&0 \, \\
\, 0&1&0&0 \,
\end{bmatrix}\!.     
\label{eq:matrices1}
\end{align}
Next, the proposed output-feedback LPV gain-scheduling control design method is described.
\section{LPV control design}
\label{sec:control}
We aim to design an output-feedback LPV gain-scheduled controller for the SPMSM model (\ref{eq:lpverrordynamics}) in the context of induced $\mathcal{L}_2$-norm performance specifications. To this end, we consider a generic LPV open-loop system with the following state-space realization  
\begin{align}
    \dot{\mathbf{x}}(t) & =  \mathbf{A}(\boldsymbol{\rho}(t)) \mathbf{x}(t)+ \mathbf{B}_1(\boldsymbol{\rho}(t))\mathbf{w}(t)+ \mathbf{B}_2(\boldsymbol{\rho}(t))\mathbf{u}(t), \nonumber\\[0.10cm] 
  \mathbf{z}(t) & = \mathbf{C}_1(\boldsymbol{\rho}(t)) \mathbf{x}(t) + \mathbf{D}_{11}(\boldsymbol{\rho}(t)) \mathbf{w}(t) + \mathbf{D}_{12}(\boldsymbol{\rho}(t)) \mathbf{u}(t),\nonumber\\[0.10cm] 
 \mathbf{y}(t)& = \mathbf{C}_2(\boldsymbol{\rho}(t)) \mathbf{x}(t)+  \mathbf{D}_{21}(\boldsymbol{\rho}(t)) \mathbf{w}(t),\nonumber\\[0.10cm] 
 \mathbf{x}(0) & = \mathbf{x}_0,
\label{LPVsystem} 
\end{align}
\noindent where $\mathbf{x} \in \mathbb{R}^n$ is the system state vector, $\mathbf{w} \in \mathbb{R}^{n_w}$ is the vector of exogenous disturbances with finite energy in the space $\mathcal{L}_2[0, \:\: \infty]$, $\mathbf{u} \in \mathbb{R}^{n_u}$ is the control input vector, $\mathbf{z}(t) \in \mathbb{R}^{n_z}$ is the vector of controlled output, $\mathbf{y}(t) \in \mathbb{R}^{n_y}$ is the vector of measured output, $\mathbf{x}_0 \in \mathbb{R}^n$ is the initial system condition. The state space matrices $\mathbf{A}(\cdot)$, $\mathbf{B}_1(\cdot)$, $\mathbf{B}_2(\cdot)$, $\mathbf{C}_1(\cdot)$, $\mathbf{C}_2(\cdot)$, $\mathbf{D}_{11}(\cdot)$, $\mathbf{D}_{12}(\cdot)$, and $\mathbf{D}_{21}(\cdot)$ have rational dependence on the time-varying scheduling parameter vector, $\boldsymbol{\rho}(\cdot) \in \mathscr{F}^\nu _\mathscr{P}$, which is also measurable in real-time. $\mathscr{F}^\nu _\mathscr{P}$ is the set of allowable parameter trajectories defined as
\begin{align}
 \!\!\!\!\!\!\!\!\!\!\!\!\!\!\!\!\!\mathscr{F}^\nu _\mathscr{P} \triangleq \{\boldsymbol{\rho}(t) \in \mathcal{C}(\mathbb{R}_{+},\mathbb{R}^{n_s}):\boldsymbol{\rho}(t) \in \mathscr{P},\qquad\quad \nonumber\\ 
  |\dot{\rho}_i (t)| \leq \nu_i, i=1,2,\dots,n_s\},
 \label{eq:parametertraj}
\end{align}
wherein $n_s$ is the number of parameters and $\mathscr{P}$ is a compact subset of $\mathbb{R}^{s}$, $i.e.$ the parameter trajectories and parameter variation rates are assumed bounded as defined. 
The output-feedback LPV gain-scheduled control design procedure consists of finding a full-order dynamic LPV controller in the form of
\begin{align}
   \dot{\mathbf{x}}_K(t) & = \mathbf{A}_K (\boldsymbol{\rho}(t)) \mathbf{x}_K(t)+\mathbf{B}_K (\boldsymbol{\rho}(t))\mathbf{y}(t),\nonumber\\[0.1cm] 
\mathbf{u}(t) & = \mathbf{C}_K (\boldsymbol{\rho}(t)) \mathbf{x}_K(t)+\mathbf{D}_K (\boldsymbol{\rho}(t))\mathbf{y}(t),
\label{controller} 
\end{align}
\noindent where $\mathbf{x}_K (t) \in \mathbb{R}^{n}$ is the controller state vector. By substituting the controller (\ref{controller}) in the open-loop system (\ref{LPVsystem}), and assuming $\mathbf{x}_{cl}(t) = [\begin{array}{cc}
\mathbf{x}(t) &\mathbf{x}_K (t)\end{array} ]^{\text{T}}$, the interconnected closed-loop system ($\mathbf{T}_{\mathbf{z}\mathbf{w}}$) is obtained as follows
\begin{align}
  \dot{\mathbf{x}}_{cl}(t)& = \mathbf{A}_{cl} (\boldsymbol{\rho}(t)) {\mathbf{x}}_{cl}(t) + \mathbf{B}_{cl} (\boldsymbol{\rho}(t)) \mathbf{w}(t),\nonumber\\[0.1cm]  
\mathbf{z}(t) & =  \mathbf{C}_{cl} (\boldsymbol{\rho}(t)) {\mathbf{x}}_{cl}(t) + \mathbf{D}_{cl} (\boldsymbol{\rho}(t)) \mathbf{w}(t),
\label{eq:closed-loop system}  
\end{align}
\noindent with
\begin{align*}
		& \mathbf{A}_{cl}=\begin{bmatrix}
\mathbf{A} + \mathbf{B}_2 \mathbf{D}_K \mathbf{C}_2 & \mathbf{B}_2 \mathbf{C}_K\\ 
\mathbf{B}_K \mathbf{C}_2 & \mathbf{A}_K 
\end{bmatrix},\\
&   \mathbf{B}_{cl}=\begin{bmatrix}
\mathbf{B}_1 + \mathbf{B}_2 \mathbf{D}_K \mathbf{D}_{21} \\
\mathbf{B}_K \mathbf{D}_{21} 
\end{bmatrix}, \\
& \mathbf{C}_{cl}=\begin{bmatrix}
\mathbf{C}_1 + \mathbf{D}_{12} \mathbf{D}_K \mathbf{C}_2 & \mathbf{D}_{12} \mathbf{C}_K
\end{bmatrix}, \\
& \mathbf{D}_{cl}= \mathbf{D}_{11} + \mathbf{D}_{12} \mathbf{D}_K \mathbf{D}_{21}, 
\end{align*}
\noindent where the dependence on the scheduling parameter has been dropped for brevity.  The final designed controller should be able to meet the following objectives for the closed-loop system:
	\begin{itemize}
		\item Input-to-state stability (ISS) of the closed-loop system (\ref{eq:closed-loop system}) in the presence of parameter variations and disturbances, and
		\item Minimization of the worst-case amplification of the induced $\mathcal{L}_2$-norm of the mapping from the disturbances $\mathbf{w}(t)$ to the controlled output $\mathbf{z}(t)$, given by
			\begin{equation}
				\Vert \mathbf{T}_{\mathbf{z}\mathbf{w}}\Vert_{i,2} = \underset{\boldsymbol{\rho}(t) \in \mathscr{F}^\nu _\mathscr{P}}{\sup} \:\:\: \underset{\Vert \mathbf{w}(t) \Vert_2 \neq 0}{\sup}\:\: \frac{\Vert \mathbf{z}(t) \Vert_2}{\Vert \mathbf{w}(t) \Vert_2}.
				\label{eq:Performance Index}
			\end{equation}
	\end{itemize}
Accordingly, in this paper, we utilize an extended form of the Bounded Real Lemma \cite{briat2014linear} and a quadratic parameter-dependent Lyapunov functions of the form $V(\mathbf{x}_{cl}(t), \boldsymbol{\rho}(t)) = \mathbf{x}^{\text{T}}_{cl}(t) \mathbf{P}(\boldsymbol{\rho}(t)) \mathbf{x}_{cl}(t)$ to obtain less conservative results that are valid for arbitrary bounded parameter variation rates \cite{apkarian1998advanced}. To this end, considering the closed-loop system (\ref{eq:closed-loop system}), the following result provides sufficient conditions for the synthesis of a output-feedback LPV controller, which is formulated as convex optimization problems with LMI constraints. The designed LPV gain-scheduled controller guarantees closed-loop asymptotic parameter-dependent quadratic (PDQ) stability and a specified performance level as defined in (\ref{eq:Performance Index}).
\begin{theorem}\label{thm:thm1}\cite{briat2014linear} Considering the given open-loop LPV system (\ref{LPVsystem}), there exists a gain-scheduled dynamic full-order output-feedback controller of the form (\ref{controller}) that guarantees the closed-loop asymptotic stability and satisfies the induced $\mathcal{L}_2$-norm performance condition  $\Vert \mathbf{z}(t) \Vert_2 \leq \gamma \Vert \mathbf{w}(t) \Vert_2$, if there exist continuously differentiable parameter-dependent symmetric matrices $\mathbf{X}, \mathbf{Y}:\mathbb{R}^{s}\rightarrow\mathbb{S}^{n}$, parameter-dependent matrices $\widehat{A} \in \mathbb{R}^{s}\rightarrow \mathbb{R}^{n \times n}$, $\widehat{B} \in \mathbb{R}^{s}\rightarrow \mathbb{R}^{n \times n_y}$, $\widehat{C} \in \mathbb{R}^{s}\rightarrow \mathbb{R}^{n_u \times n}$, $\widehat{D} \in \mathbb{R}^{s}\rightarrow \mathbb{R}^{n_u \times n_y}$, and a scalar $\gamma > 0$ such that the following LMI conditions hold for all $\boldsymbol{\rho} \in \mathscr{F}^\nu _\mathscr{P}$.
\begin{equation}
\begin{array}{l}
\left[\begin{array}{cc}
\dot{\mathbf{X}} + \mathbf{X} \mathbf{A} + \widehat{B}\mathbf{C}_2 + (\star)   & \star \\
\widehat{A}^{\text{T}} + \mathbf{A} + \mathbf{B}_2 \widehat{D} \mathbf{C}_2 & -\dot{\mathbf{Y}} + \mathbf{A}\mathbf{Y} + \mathbf{B}_2 \widehat{C} + (\star) \\
(\mathbf{X} \mathbf{B}_1 + \widehat{B} \mathbf{D}_{21})^{\text{T}} & (\mathbf{B}_1 + \mathbf{B}_2 \widehat{D} \mathbf{D}_{21})^{\text{T}} \\
\mathbf{C}_1 + \mathbf{D}_{12} \widehat{D} \mathbf{C}_2 & \mathbf{C}_1 \mathbf{Y} + \mathbf{D}_{12} \widehat{C}
\end{array}\right.\\[8mm]
\quad \qquad\qquad\quad\quad\quad\qquad\left.\begin{array}{cc}
\star & \star \\
\star & \star \\
 -\gamma \mathbf{I}   & \star \\
\mathbf{D}_{11} + \mathbf{D}_{12} \widehat{D} \mathbf{D}_{21} &  -\gamma \mathbf{I}
\end{array}\right]\prec\mathbf{0},
\end{array}
\label{eq:LMI1}
\end{equation}
\begin{equation}
\left[\begin{array}{cc}
\mathbf{X} & \mathbf{I} \\
\mathbf{I} & \mathbf{Y}
\end{array}\right] \succ	 \mathbf{0}.
\label{eq:LMI2}
\end{equation}
\end{theorem}
Subsequently, the LPV control design is expanded to guarantee robustness against modeling mismatch and parameter uncertainties. To this end, $\mathbf{A}$ and $\mathbf{B}_2$ in (\ref{LPVsystem}) are considered to be uncertain system matrices,  $\mathbf{A}_{\Delta}(\boldsymbol{\rho}(t)) = \mathbf{A}(\boldsymbol{\rho}(t))  + \boldsymbol{\Delta} \mathbf{A}(t)$, $\mathbf{B}_{2, \Delta}(\boldsymbol{\rho}(t)) = \mathbf{B}_2(\boldsymbol{\rho}(t))  + \boldsymbol{\Delta} \mathbf{B}_2(t)$, where $\boldsymbol{\Delta} \mathbf{A}(t)$ and $\boldsymbol{\Delta} \mathbf{B}_2(t)$  are bounded matrices containing parametric uncertainties. The norm-bounded uncertainties are assumed to satisfy the following relation
\begin{equation}
\left[\begin{array}{c}
\boldsymbol{\Delta} \mathbf{A}(t)\\
\boldsymbol{\Delta} \mathbf{B}_2(t) 
\end{array} 
\right] = \mathbf{H} \boldsymbol{\Delta} (t) \left[ \begin{array}{c}
\mathbf{E}_1\\
\mathbf{E}_2\end{array} \right],
\label{uncertainties}
 \end{equation} 
\noindent where $\mathbf{H} \in \mathbb{R}^{n \times i} $, $\mathbf{E}_1 \in \mathbb{R}^{j \times n} $, $\mathbf{E}_2 \in \mathbb{R}^{j \times n_u} $ are known constant matrices and $\boldsymbol{\Delta} (t) \in \mathbb{R}^{i \times j}$ is an unknown  time-varying uncertainty matrix function satisfying inequality
\begin{equation}
 \boldsymbol{\Delta}^T (t) \boldsymbol{\Delta} (t)  \preceq \mathbf{I}.
\label{Delta1}
\end{equation}
By substituting $\mathbf{A}_{\Delta}(\boldsymbol{\rho}(t))$ and $\mathbf{B}_{2, \Delta}(\boldsymbol{\rho}(t))$ for $\mathbf{A}$ and $\mathbf{B}_2$ in (\ref{eq:LMI1}), the following result presents a condition for ensuring closed-loop stability and performance in the presence of norm-bounded uncertainties via an LPV control design of the form (\ref{controller}).
\begin{theorem}\label{thm:thm2} There exists a full-order robust output-feedback LPV controller of the form (\ref{controller}), over the sets $\mathscr{F}^\nu _\mathscr{P}$  with all admissible uncertainties $\boldsymbol{\Delta} \mathbf{A}(t)$ and $\boldsymbol{\Delta} \mathbf{B}_2(t)$ of the form (\ref{uncertainties}) and all $\boldsymbol{\Delta}(t)$ satisfying (\ref{Delta1}), that guarantees the closed-loop asymptotic stability and satisfies the induced $\mathcal{L}_2$-norm performance condition  $\Vert \mathbf{z}(t) \Vert_2 \leq \gamma \Vert \mathbf{w}(t) \Vert_2$, if there exist continuously differentiable parameter dependent symmetric matrices $\mathbf{X}, \mathbf{Y}:\mathbb{R}^{s}\rightarrow\mathbb{S}^{n}$, parameter dependent real matrices $\widehat{A} \in \mathbb{R}^{s}\rightarrow \mathbb{R}^{n \times n}$, $\widehat{B} \in \mathbb{R}^{s}\rightarrow \mathbb{R}^{n \times n_y}$, $\widehat{C} \in \mathbb{R}^{s}\rightarrow \mathbb{R}^{n_u \times n}$, $\widehat{D} \in \mathbb{R}^{s}\rightarrow \mathbb{R}^{n_u \times n_y}$, and a positve scalars $\gamma$, and $\boldsymbol{\epsilon}$ such that the LMI (\ref{eq:robustLMIClosedloop}) is feasible.
\begin{figure*}[!t]
\begin{equation}
\left[\begin{array}{cccccccc}
\!\dot{\mathbf{X}} \!+\! \mathbf{X} \mathbf{A} + \widehat{B}\mathbf{C}_2 \!+\! (\star)   & \star & \star & \star & \star & \star & \star & \star \\[.05cm]
\widehat{A}^{\text{T}} + \mathbf{A} + \mathbf{B}_2 \widehat{D} \mathbf{C}_2 & -\dot{\mathbf{Y}} + \mathbf{A}\mathbf{Y} + \mathbf{B}_2 \widehat{C} + (\star) & \star & \star & \star & \star & \star & \star \\[.05cm]
(\mathbf{X} \mathbf{B}_1 + \widehat{B} \mathbf{D}_{21})^{\text{T}} & (\mathbf{B}_1 + \mathbf{B}_2 \widehat{D} \mathbf{D}_{21})^{\text{T}} &  -\gamma \mathbf{I} & \star & \star & \star & \star & \star \\[.05cm]
\mathbf{C}_1 + \mathbf{D}_{12} \widehat{D} \mathbf{C}_2 & \mathbf{C}_1 \mathbf{Y} + \mathbf{D}_{12} \widehat{C} & \mathbf{D}_{11} + \mathbf{D}_{12} \widehat{D} \mathbf{D}_{21} &  -\gamma \mathbf{I} & \star & \star & \star & \star \\[.05cm]
\mathbf{H}^\text{T} \mathbf{X} & \mathbf{H}^\text{T} & \mathbf{0} & \mathbf{0} & -\boldsymbol{\epsilon} \mathbf{I} & \star & \star & \star \\[.05cm]
\boldsymbol{\epsilon} \mathbf{E}_1 & \boldsymbol{\epsilon} \mathbf{E}_1 \mathbf{Y} & \mathbf{0}& \mathbf{0} & \mathbf{0} & -\boldsymbol{\epsilon} \mathbf{I} & \star & \star \\[.05cm]
\mathbf{0} & \mathbf{H}^\text{T} & \mathbf{0} & \mathbf{0} & \mathbf{0} & \mathbf{0} & -\boldsymbol{\epsilon} \mathbf{I} & \star \\[.05cm]
\boldsymbol{\epsilon} \mathbf{E}_2 \widehat{D} \mathbf{C}_2 & \boldsymbol{\epsilon} \mathbf{E}_2 \widehat{C} & \boldsymbol{\epsilon} \mathbf{E}_2 \widehat{D} \mathbf{D}_{21}  & \mathbf{0} & \mathbf{0} & \mathbf{0} & \mathbf{0} & -\boldsymbol{\epsilon} \mathbf{I}
\end{array}\right] \prec \mathbf{0}
\label{eq:robustLMIClosedloop}
\end{equation}
\end{figure*}

\end{theorem}
\begin{proof} 
By substituting the matrices with additive norm-bounded uncertainties in the LMI condition (\ref{eq:LMI1}) given by Theorem \ref{thm:thm1}, \textit{i.e.}, $\mathbf{A}_{\Delta}(\boldsymbol{\rho}(t))$  for $\mathbf{A}$ and $\mathbf{B}_{2, \Delta}(\boldsymbol{\rho}(t))$ for $\mathbf{B}_2$ in (\ref{eq:LMI1}), the new LMI condition will be as follows
\begin{equation}
\begin{array}{l}
(\ref{eq:LMI1}) + 
\mathbf{He}\Bigg(
\left[\begin{array}{c}
\mathbf{X}\mathbf{H}\\
\mathbf{H}\\
\mathbf{0}\\
\mathbf{0}\end{array}\right]
\boldsymbol{\Delta}(t)
\left[\begin{array}{cccc}
\mathbf{E}_1 & \mathbf{E}_1\mathbf{Y} & \mathbf{0} & \mathbf{0}
\end{array}\right]\Bigg)\\[.8cm]
\!+\mathbf{He}\Bigg(\!
\left[\!\!\begin{array}{c}
\mathbf{0}\\
\mathbf{H}\\
\mathbf{0}\\
\mathbf{0}\end{array}\!\!\right]
\!\boldsymbol{\Delta}(t)\!
\left[\!\!\begin{array}{cccc}
\mathbf{E}_2 \widehat{D}\mathbf{C}_2 & \!\mathbf{E}_2\widehat{C} & \!\mathbf{E}_2 \widehat{D} \mathbf{D}_{21} & \mathbf{0}
\end{array}\!\!\right]\!\!\Bigg)
\!\!\prec\!\mathbf{0}.
\end{array}
\end{equation}
\noindent Finally, using the following inequality \cite{xie1996output} 
\begin{equation}
\boldsymbol{\Theta} \boldsymbol{\Delta} (t) \boldsymbol{\Phi} + \boldsymbol{\Phi}^{\text{T}} \boldsymbol{\Delta}^{\text{T}} (t) \boldsymbol{\Theta}^{\text{T}} \leq \boldsymbol{\epsilon}^{-1} \boldsymbol{\Theta} \boldsymbol{\Theta}^{\text{T}} + \boldsymbol{\epsilon} \boldsymbol{\Phi} ^{\text{T}} \boldsymbol{\Phi},
\label{eq:ineq}
\end{equation}
\noindent which holds for all scalars $\boldsymbol{\epsilon}>0$ and all constant matrices $\boldsymbol{\Theta}$ and $\boldsymbol{\Phi}$ of appropriate dimensions, and using the Schur complement \cite{boyd1994linear}, the final LMI condition (\ref{eq:robustLMIClosedloop}) is obtained. 
\end{proof}
% \end{theorem}
Once the parameter-dependent LMI decision matrices, $\mathbf{X}$, $\mathbf{Y}$, $\widehat{A}$, $\widehat{B}$, $\widehat{C}$, and $\widehat{D}$  satisfying the LMI conditions (\ref{eq:LMI1}) and (\ref{eq:LMI2}) are obtained, the output-feedback LPV controller matrices can be readily computed following the steps below:

1. Determine $\mathbf{M}$ and $\mathbf{N}$ from the factorization problem
\begin{equation}
\mathbf{I} - \mathbf{X}\mathbf{Y} = \mathbf{N} \mathbf{M}^{\text{T}},
\end{equation} 
\noindent where the obtained $\mathbf{M}$ and $\mathbf{N}$ matrices are square and invertible in the case of a full-order controller.
% \noindent 2. Compute the following parameter matrices:
% \begin{equation}
% \begin{array}{cl}
% \widehat{A}&= \mathbf{X} \mathbf{A} \mathbf{Y} + \mathbf{X} \mathbf{B}_2 \mathbf{D}_k \mathbf{C}_2 \mathbf{Y} + \mathbf{N} \mathbf{B}_k \mathbf{C}_2 \mathbf{Y}\\
% & + \mathbf{X} \mathbf{B}_2 \mathbf{C}_k \mathbf{M}^{\text{T}} + \mathbf{N} \mathbf{A}_k \mathbf{M}^{\text{T}},\\[0.2cm]
% \widehat{A}_d&= \mathbf{X} \mathbf{A}_d \mathbf{Y} + \mathbf{X} \mathbf{B}_2 \mathbf{D}_k \mathbf{C}_{2d} \mathbf{Y} + \mathbf{N} \mathbf{B}_k \mathbf{C}_{2d} \mathbf{Y}\\
% & + \mathbf{X} \mathbf{B}_2 \mathbf{C}_{dk} \mathbf{M}^{\text{T}} + \mathbf{N} \mathbf{A}_{dk} \mathbf{M}^{\text{T}},\\[0.2cm]
% \widehat{B}&= \mathbf{X} \mathbf{B}_2 \mathbf{D}_k + \mathbf{N} \mathbf{B}_k,\\
% \widehat{C}&= \mathbf{D}_k \mathbf{C}_2 \mathbf{Y} + \mathbf{C}_k \mathbf{M}^{\text{T}},\\
% \widehat{C}_d&= \mathbf{D}_k \mathbf{C}_{2d} \mathbf{Y} + \mathbf{C}_{dk} \mathbf{M}^{\text{T}}.
% \end{array}
% \end{equation}

2. Compute the controller matrices in the following order:
\begin{align}
  \mathbf{D}_K&=\widehat{D},\nonumber\\
\mathbf{C}_K&= (\widehat{C} - \mathbf{D}_K \mathbf{C}_{2} \mathbf{Y}) \mathbf{M} ^{-\text{T}},\nonumber\\
\mathbf{B}_K&= \mathbf{N}^{-1} (\widehat{B} -  \mathbf{X} \mathbf{B}_2 \mathbf{D}_K),\nonumber\\
\mathbf{A}_{K}&= -\mathbf{N}^{-1} (\mathbf{X} \mathbf{A} \mathbf{Y} + \mathbf{X} \mathbf{B}_2 \mathbf{D}_K \mathbf{C}_{2} \mathbf{Y} + \mathbf{N} \mathbf{B}_K \mathbf{C}_{2} \mathbf{Y}\nonumber\\
&+ \mathbf{X} \mathbf{B}_2 \mathbf{C}_{K} \mathbf{M} ^{\text{T}} - \widehat{A}) \mathbf{M} ^{-\text{T}}.
\label{eq:controllermatrices}  
\end{align}
\begin{remark}
\label{remark1}
Theorem \ref{thm:thm1} results in an infinite-dimensional convex optimization problem with an infinite number of LMIs and decision variables since the scheduling parameter vector belongs to a continuous real vector space, $\boldsymbol{\rho} \in \mathscr{F}^\nu _\mathscr{P}$. To address this obstacle, the gridding method of the parameter space is utilized to convert the infinite-dimensional problem to a finite-dimensional convex optimization problem \cite{apkarian1998advanced}. In this regard, we choose the matrix parameter functional dependence as $\mathbf{M}(\boldsymbol{\rho}(t))=\mathbf{M}_0 + \sum\limits_{i=1}^{s}\rho_i(t) \mathbf{M}_{i_1}$, where $\mathbf{M}(\boldsymbol{\rho}(t))$ represents any of the parameter-dependent matrices appearing in the LMI conditions (\ref{eq:LMI1}), and (\ref{eq:LMI2}). Subsequently, by gridding the scheduling parameter space at appropriate intervals we obtain a finite set of LMIs to be solved for the unknown matrices and $\gamma$. The MATLAB\textsuperscript{\tiny\textregistered} toolbox YALMIP can be used to solve the introduced optimization problem \cite{lofberg2004yalmip}. Also,  it should be noted that due to the presence of derivatives of the parameter-dependent matrices in the LMI condition  (\ref{eq:LMI1}), \textit{i.e.}  $\dot{\mathbf{X}}$, and $\dot{\mathbf{Y}}$, the parameter variation rate $\dot{\rho}$, enters affinely in the LMIs, and it is sufficient to check the LMI only at the vertices of the $\dot{\rho}$ parameter range. 
\end{remark}
%In the next section, the problem of controlling the AFR in SI engines with TWC will be introduced as a numerical example to assess the performance of the output-feedback controller synthesis conditions presented in this section.
% more notes need to be added here regarding the application and so on.
% \noindent where $(1,1)=\mathbf{A}^{\text{T}}\mathbf{P}+ \mathbf{P}\mathbf{A}+ \dfrac{\partial \mathbf{P}}{\partial \rho_i}\dot{\rho}_i  + \mathbf{Q}-\mathbf{R}$, and $\boldsymbol{\Xi}_{22}=- \big(1-   \dfrac{\partial \tau}{\partial \rho_i} \dot{\rho}_i \big)  \mathbf{Q}-\mathbf{R}$ and and the dependence on the scheduling parameter has been dropped for the clarity and simplicity. It should be noted that the parameter variation rate $\dot{\rho}$, enters affinely in the LMIs in this paper, hence, it is enough to check LMI only at the vertices of $\dot{\rho}$. As a result, $\dfrac{\partial \mathbf{P}}{\partial \rho_i} \dot{\rho}_i $ and $\dfrac{\partial \tau}{\partial \rho_i} \dot{\rho}_i $ are replaced by $\sum_{i=1}^s \pm \bigg(\nu_i \dfrac{\partial \mathbf{P}}{\partial \rho_i}\bigg)$  and $\sum_{i=1}^s \pm \bigg(\nu_i \dfrac{\partial \tau}{\partial \rho_i} \bigg)$, respectively.
\section{SPMSM LPV control design}
\label{sec:results}
 We examine the application of the proposed LPV gain-scheduled control design method to the SPMSM speed regulation. The SPMSM error dynamics (\ref{eq:lpverrordynamics}) is formulated in an  LPV framework (shown in Section \textcolor{red}{III}) which is suitable for the proposed LPV control design synthesis.  Considering the generic LPV system state-space realization (\ref{LPVsystem}) for the SPMSM model described in (\ref{eq:lpverrordynamics}), the LPV state-space matrices of the SPMSM are as shown in (\ref{eq:matrices1}). Moreover, the vector of the target outputs to be controlled is defined as follows
\begin{equation}
\begin{array}{l}
\!\!\mathbf{z}^{\!\text{T}}(t) = \left[\!\!\begin{array}{ccc}
\phi \cdot e_z (t)&
\sigma \cdot e_\omega (t) &
\xi \cdot e_\alpha (t)\end{array}\right.\\[.1cm]
\qquad\qquad\quad\quad\left.\begin{array}{ccc}
\psi \cdot e_\beta (t)&
\eta \cdot u_\alpha (t)&
\mu \cdot u_\beta (t)
\end{array}\!\!\right].
\end{array}
\end{equation}
% \begin{equation}
% \mathbf{z}(t) = \begin{bmatrix}
% \phi \cdot e_z (t)\\
% \sigma \cdot e_\omega (t) \\
% \xi \cdot e_\alpha (t)\\
% \psi \cdot e_\beta (t)\\
% \eta \cdot u_\alpha(t)\\ 
% \mu \cdot u_\beta (t)
% \end{bmatrix} = \mathbf{C}_{1} \mathbf{e}(t) + \mathbf{D}_{11} \mathbf{w}(t) + \mathbf{D}_{12} \mathbf{u}(t),
% \end{equation}
% \noindent with 
% \begin{equation*}
%     \mathbf{C}_{1}=\begin{bmatrix}
% \phi & 0 & 0 & 0\\
% 0 & \sigma & 0 & 0\\
% 0 & 0 & \xi & 0\\
% 0 & 0 & 0 & \psi\\
% 0 & 0 & 0 & 0\\
% 0 & 0 & 0 & 0
% \end{bmatrix}, \mathbf{D}_{12}=\begin{bmatrix}
% 0 & 0\\
% 0 & 0\\
% 0 & 0\\
% 0 & 0\\
% \eta & 0\\
% 0 & \mu
% \end{bmatrix},
% \end{equation*}
\noindent The velocity tracking error which is included in the second state $x_2(t) = e_w(t)$ is penalized by the weighting scalar $\sigma$ and the control efforts $u_\alpha(t)$, $u_\beta(t)$ are penalized by the weighting scalars $\eta$ and $\mu$, respectively. The choice of the weighting scalars $\phi$, $\sigma$, $\xi$, $\psi$, $\eta$, and $\mu$ determine the relative weighting in the optimization scheme and depends on the desired performance objectives that the designer seeks to achieve \cite{lee2015h}.  Now, based on the definition of the desired controlled vector $\mathbf{z}(t)$, the output-feedback controller is designed for the SPMSM to minimize the induced $\mathcal{L}_2$ gain (or $\mathcal{H}_{\infty}$ norm) (\ref{eq:Performance Index}) of the closed-loop LPV system (\ref{eq:closed-loop system}). The design objective is to guarantee closed-loop stability and minimize the worst case disturbance amplification over the entire range of model parameter variations. 
% \begin{figure}[t]
% \includegraphics[width=\columnwidth, height=1.5in]{Blockdiagram.png}   % \includegraphics[width=0.7\columnwidth, height=2.8in]{Psi2.jpg}
% \caption{Block diagram of the SPMSM control system}
% \label{fig:1}
% \end{figure}

In order to demonstrate the improved performance of the proposed control with respect to the desired velocity profile tracking and load torque disturbance rejection, closed-loop simulations are performed in the MATLAB/Simulink environment. 
% The proposed SPMSM control system with the output feedback controller is constructed as seen in Fig. \ref{fig:1}, and 
The model parameters of the SPMSM are listed in Table \ref{table: MotorParameters}. For comparison purposes, we evaluate the closed-loop tracking performance of the proposed controller against the FOC method with fixed gains. The FOC tuned PI controller transfer functions are selected as follows \cite{kim2017electric}:

\begin{table}[]
\small\sf\centering
\caption{Parameters of the SPMSM.\label{table: MotorParameters}}
\begin{tabular}{lll}
\toprule
Parameter&Value&Unit\\
\midrule
Number of pole pairs ($p$)   & 4      & -         \\
Stator resistance ($R_{s}$)  & 0.2    & $\Omega$\ \\
Stator inductance ($L_{s}$)  & 0.4    & mH        \\
Magnetic flux linkage ($\lambda_{pm}$) & 16.3   & mWb    \\
Moment of inertia ($J_{m}$)   & 3.24 $\times$ 10$^{-5}$ & kg $\cdot$m$^2$    \\
Coefficient of friction ($B$)    & 0.004   & N$\cdot$ m $\cdot$ s/rad    \\
\bottomrule
\end{tabular}
\end{table}

\begin{equation}
    \begin{split}
      G_{cs} (s)& = 0.533 + \dfrac{61.4}{s}, \\
G_{cc} (s)& = 1.38 + \dfrac{691}{s},    
    \end{split}
\end{equation}

\noindent where $G_{cs} (s)$ and $G_{cc} (s)$  indicate the speed controller in the $q$ axis and the current controllers in the $d$ and $q$ axis respectively. The gains of these controllers are obtained based on the nominal parameters of the motor and the desired bandwidth of the controllers.

In order to evaluate the closed-loop tracking performance of the proposed LPV method, we consider a desired velocity reference profile, as shown in Figure \ref{fig:Reference_step}. Figures \ref{fig:trackingstep1} and \ref{fig:trackingstep2} present the magnified plots of the tracking error result of the LPV and PI controllers in the absence of any disturbances, for the first step change, when the velocity reference accelerates from 0 r/min to 300 r/min, and for the second step change, when the velocity reference decelerates from 300 r/min to 100 r/min, respectively. As anticipated, the proposed LPV controller outperforms the PI controller with respect to the overshoot/undershoot, rise time, and speed of the response in both acceleration and deceleration intervals due to its scheduling structure. Figures \ref{fig:nodistcurrent} and \ref{fig:nodistvoltage} show the currents and control input voltages of the proposed LPV controller, both in the $\alpha - \beta$ axis.

\begin{figure}[t]
\centering
\includegraphics[width=\columnwidth, height=2.1in]{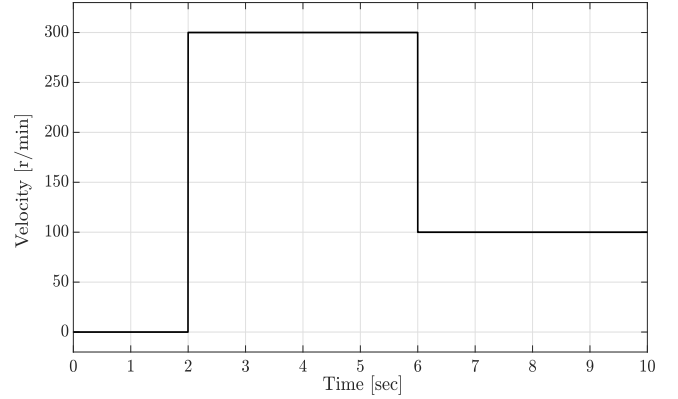}   % \includegraphics[width=0.7\columnwidth, height=2.8in]{Psi2.jpg}
\caption{Desired velocity reference profile.}
\label{fig:Reference_step}
\end{figure}
% \vspace{-5mm}
\begin{figure}[t]
\centering
\includegraphics[width=\columnwidth, height=2.1in]{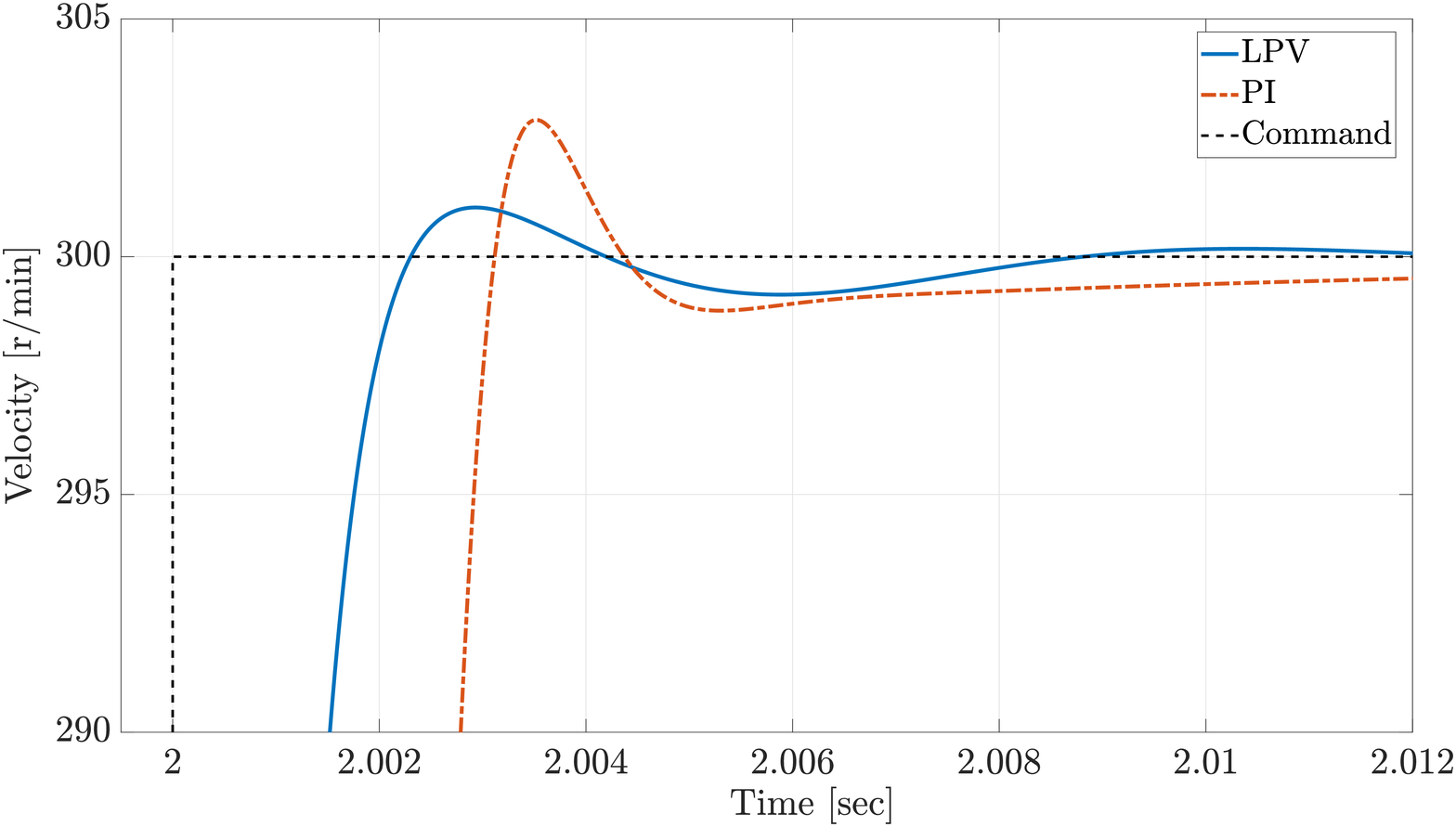}   % \includegraphics[width=0.7\columnwidth, height=2.8in]{Psi2.jpg}
\caption{Closed-loop velocity tracking performance of the LPV controller and the fixed structure PI controller with no disturbance during acceleration period.}
\label{fig:trackingstep1}
\end{figure}
% \vspace{-5mm}
\begin{figure}[t]
\centering
\includegraphics[width=\columnwidth, height=2.2in]{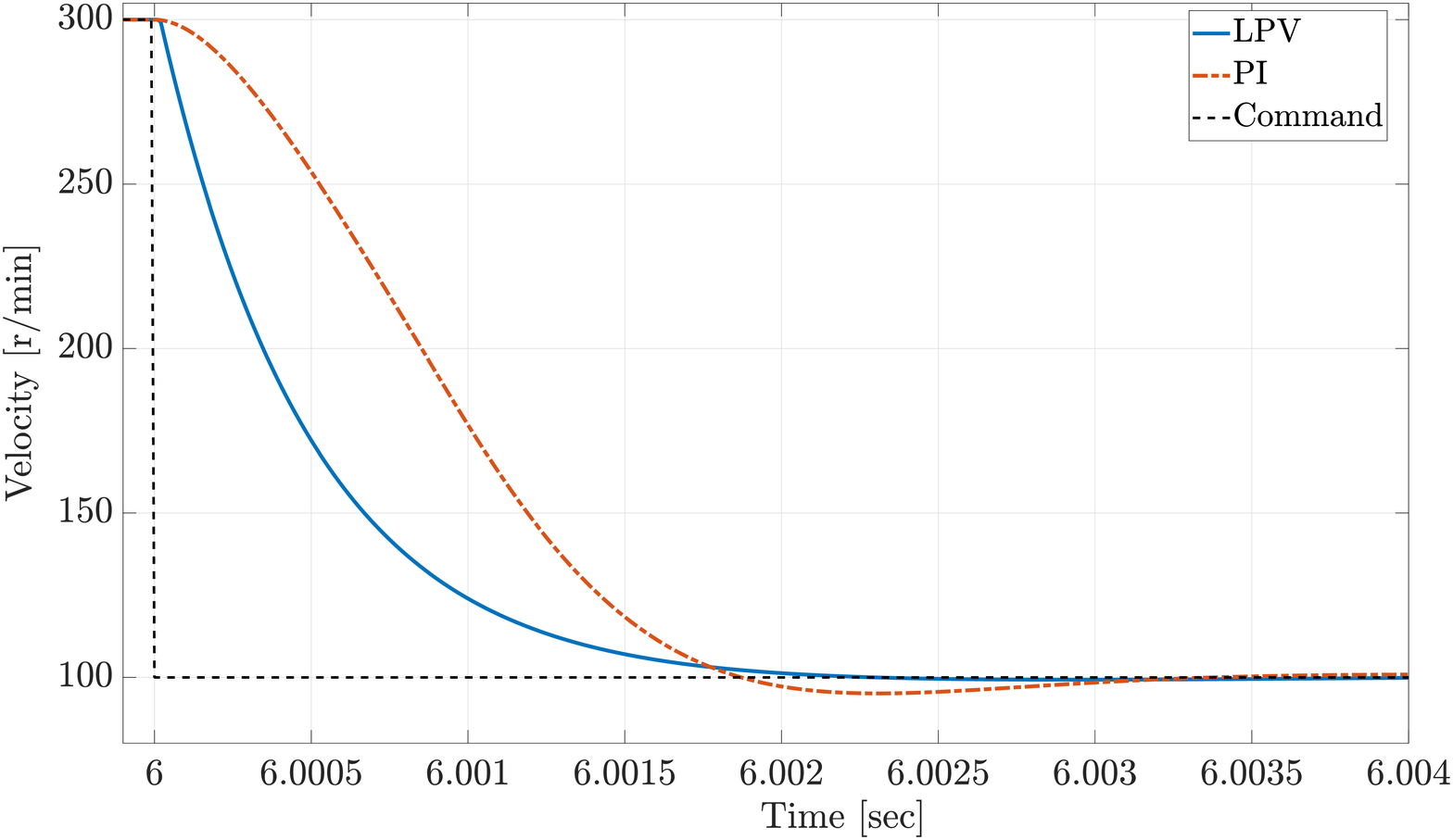}   % \includegraphics[width=0.7\columnwidth, height=2.8in]{Psi2.jpg}
\caption{Closed-loop velocity tracking performance of the LPV controller and the fixed structure PI controller with no disturbance during deceleration period.}
\label{fig:trackingstep2}
\end{figure}
% \vspace{-1mm}

\begin{figure}[t]
\centering
\includegraphics[width=\columnwidth, height=2.2in]{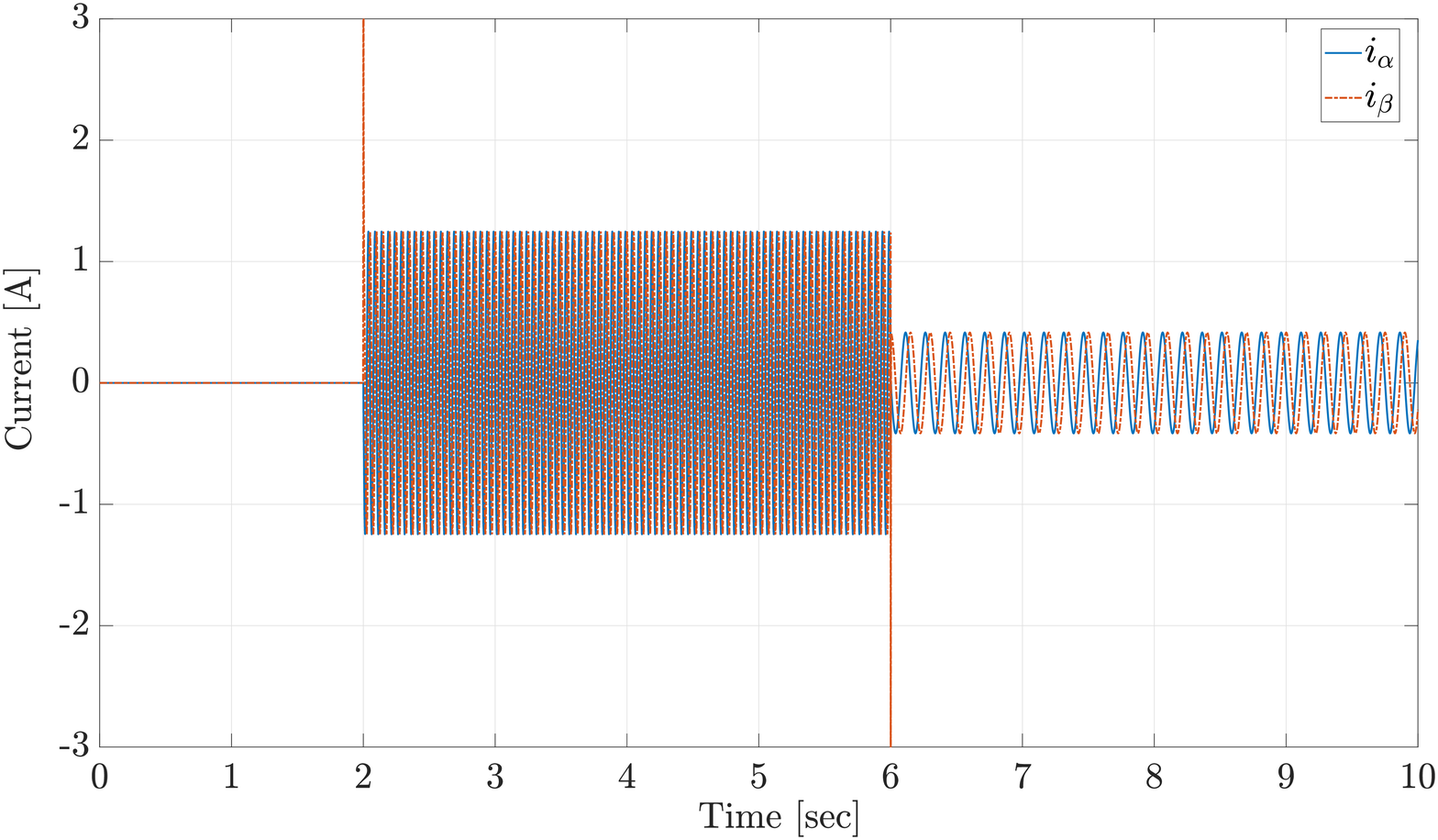}   % \includegraphics[width=0.7\columnwidth, height=2.8in]{Psi2.jpg}
\caption{Currents of the LPV controller in the $\alpha - \beta$ axis.}
\label{fig:nodistcurrent}
\end{figure}
\begin{figure}[t]
\centering
\includegraphics[width=\columnwidth, height=2.2in]{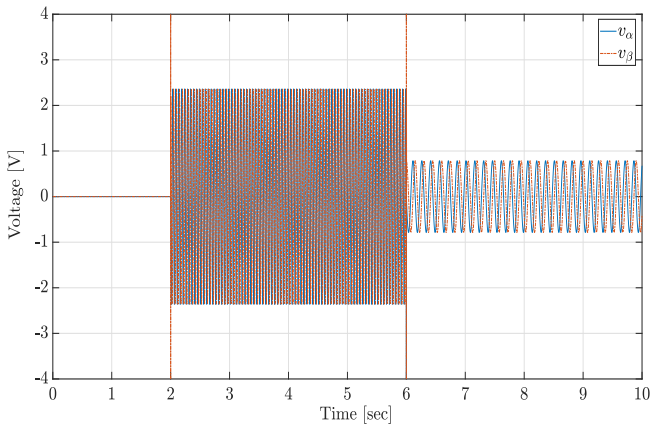}   % \includegraphics[width=0.7\columnwidth, height=2.8in]{Psi2.jpg}
\caption{Control input voltages of the LPV controller in the $\alpha - \beta$ axis.}
\label{fig:nodistvoltage}
\end{figure}
Next, we assume that the SPMSM is experiencing temperature variation with a temperature profile shown in Figure \ref{fig:temp} and an output disturbance. The temperature variation affects the model's resistance and magnet flux as described in (\ref{9}). The disturbance under consideration is a constant torque load disturbance as shown in Figure \ref{fig:dist}. The closed-loop performance of the proposed LPV controller and the PI controller in tracking a given ramp-type velocity reference command with a step disturbance is shown in Figure \ref{fig:trackingdist}. Additionally, Figures \ref{fig:distcurrent} and \ref{fig:distvoltage} demonstrate the currents and the input voltages of the LPV controller in the $\alpha - \beta$ phases, respectively.
In order to evaluate the robustness of the proposed design,  the closed-loop response of the proposed robust LPV gain-scheduling controller is investigated in the presence of model parameter variations. To this end, we select the stator inductance $L_s$  and the moment of inertia $J_m$ to be under-estimated by $50 \%$, and the stator resistance $R_s$ and viscous friction coefficient $B$ to be over-estimated by $50 \%$, which corresponds to a worst-case perturbation scenario. The closed-loop velocity tracking performance of the system with the proposed robust LPV control design (obtained through condition (\ref{eq:robustLMIClosedloop}) and Theorem \ref{thm:thm2}) is compared to the response of the LPV controller designed without considering uncertainty obtained using the results of Theorem \ref{thm:thm1}. As per Figure \ref{fig:trackingdist_robust}, the control without considering uncertainty in the design demonstrates significant oscillatory behavior, higher overshoots and settling time, which are undesirable. Hence, as the results demonstrate, the proposed robust LPV control design is capable of compensating for parameter uncertainties and modeling mismatches. Therefore, by investigating the presented results, we conclude that the proposed LPV control method demonstrates superior results in terms of velocity tracking, disturbance rejection and robustness under different simulated scenarios in the presence of parameter variations, disturbances and model uncertainty.
\begin{figure}[t]
\centering
\includegraphics[width=\columnwidth, height=2.2in]{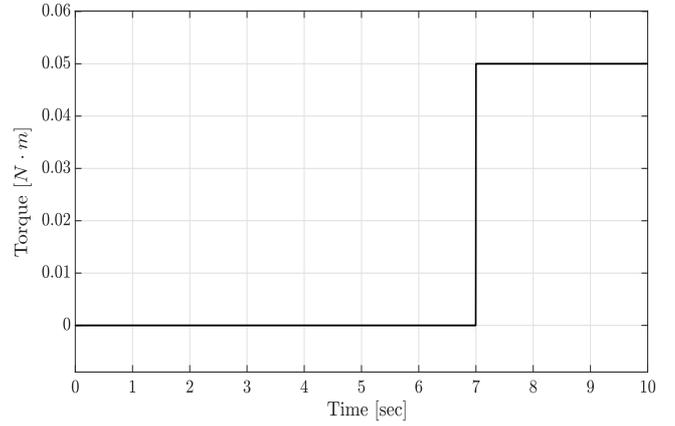}   % \includegraphics[width=0.7\columnwidth, height=2.8in]{Psi2.jpg}
\caption{Load torque disturbance.}
\label{fig:dist}
\end{figure}
\begin{figure}[t]
\centering
\includegraphics[width=\columnwidth, height=2.2in]{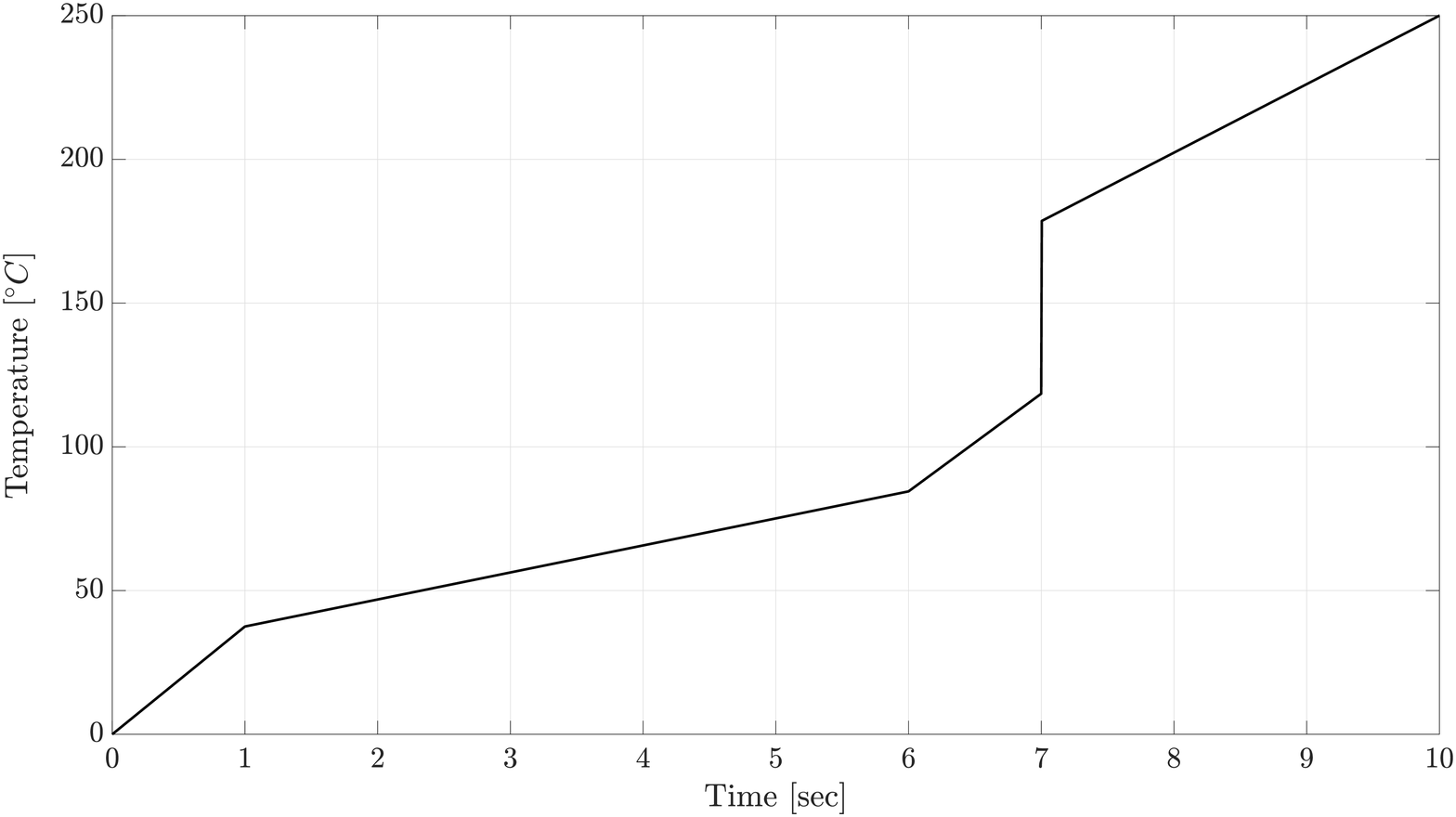}   % \includegraphics[width=0.7\columnwidth, height=2.8in]{Psi2.jpg}
\caption{SPMSM operating temperature variation.}
\label{fig:temp}
\end{figure}

\section{Conclusion}
\label{sec:conclusion}
In the present paper, a linear parameter-varying (LPV) gain-scheduled output feedback controller has been proposed for the speed control of the surface permanent magnet synchronous motors (SPMSMs). The dynamic model of the motor has been developed in the $\alpha - \beta$ stationary reference frame, and an LPV model representation has been utilized to capture the nonlinear SPMSM dynamics. The effect of temperature on the variability of SPMSMs model parameters is taken into account in the model. The linear matrix inequality (LMI) framework has been used to formulate the controller synthesis conditions as  numerically tractable convex optimization computational problem. Subsequently, the proposed controller was designed to guarantee the closed-loop stability and  minimize the disturbance amplification in terms of the induced $\mathcal{L}_2$-norm performance specification of the closed-loop system. The effectiveness of the proposed controller was validated via comparisons with a conventional PI controller in the MATLAB/Simulink environment. The results demonstrated the effectiveness and superiority of the proposed approach in improving the transient performances in terms of settling time, overshoot, disturbance attenuation, and  parametric uncertainty compensation. Future research would focus on designing a disturbance observer to empower the control design to cope better with unknown disturbances. Hence, the estimated disturbance can be compensated in the controller output to improve the stability of the motor and speed tracking performance.

\begin{figure}[t]
\centering
\includegraphics[width=\columnwidth, height=2.4in]{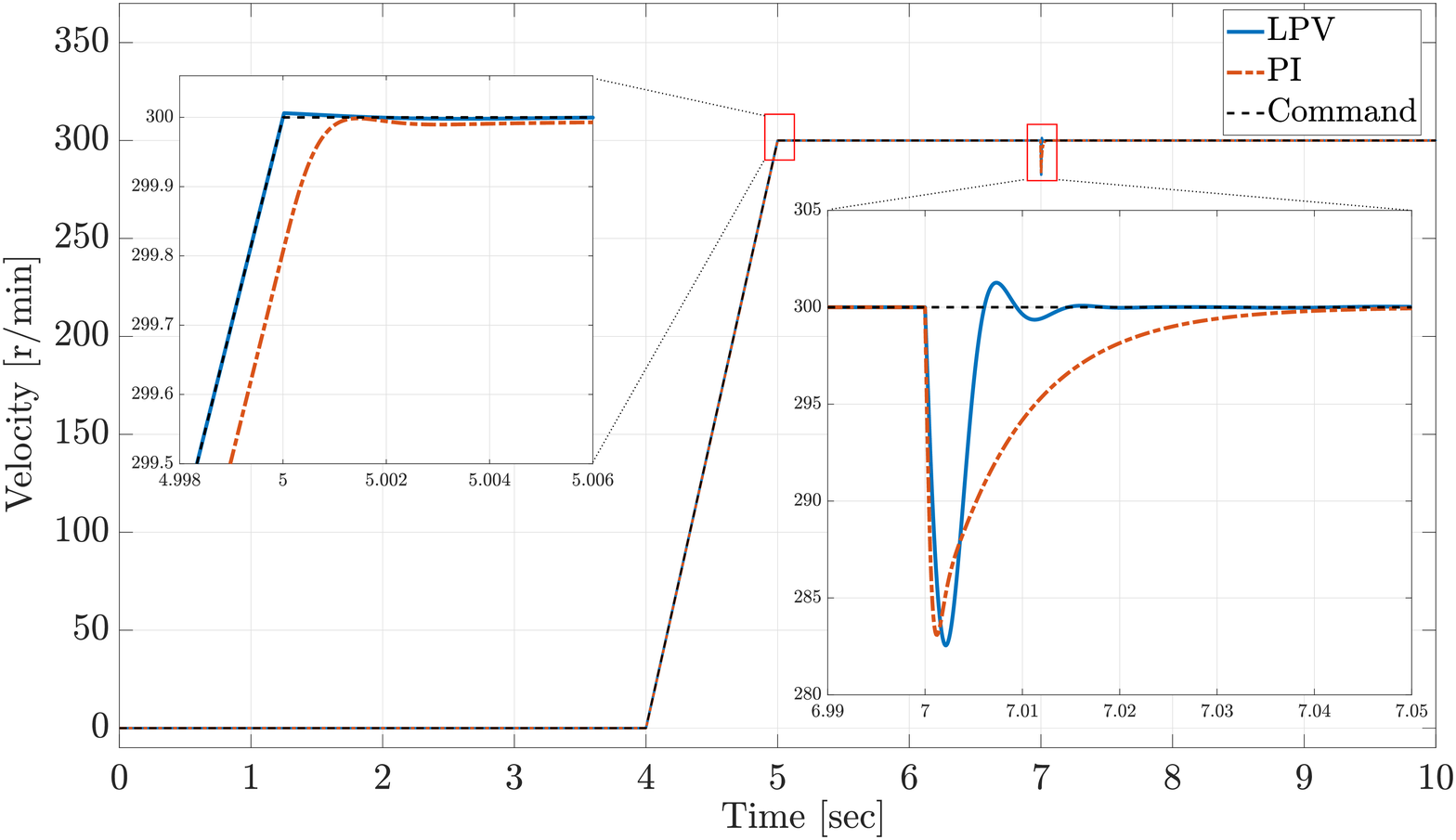}   % \includegraphics[width=0.7\columnwidth, height=2.8in]{Psi2.jpg}
\caption{Closed-loop velocity tracking performance of the LPV controller and the fixed structure PI controller subject to load torque disturbance.}
\label{fig:trackingdist}
\end{figure}
\begin{figure}[t]
\centering
\includegraphics[width=\columnwidth,height=2.3in]{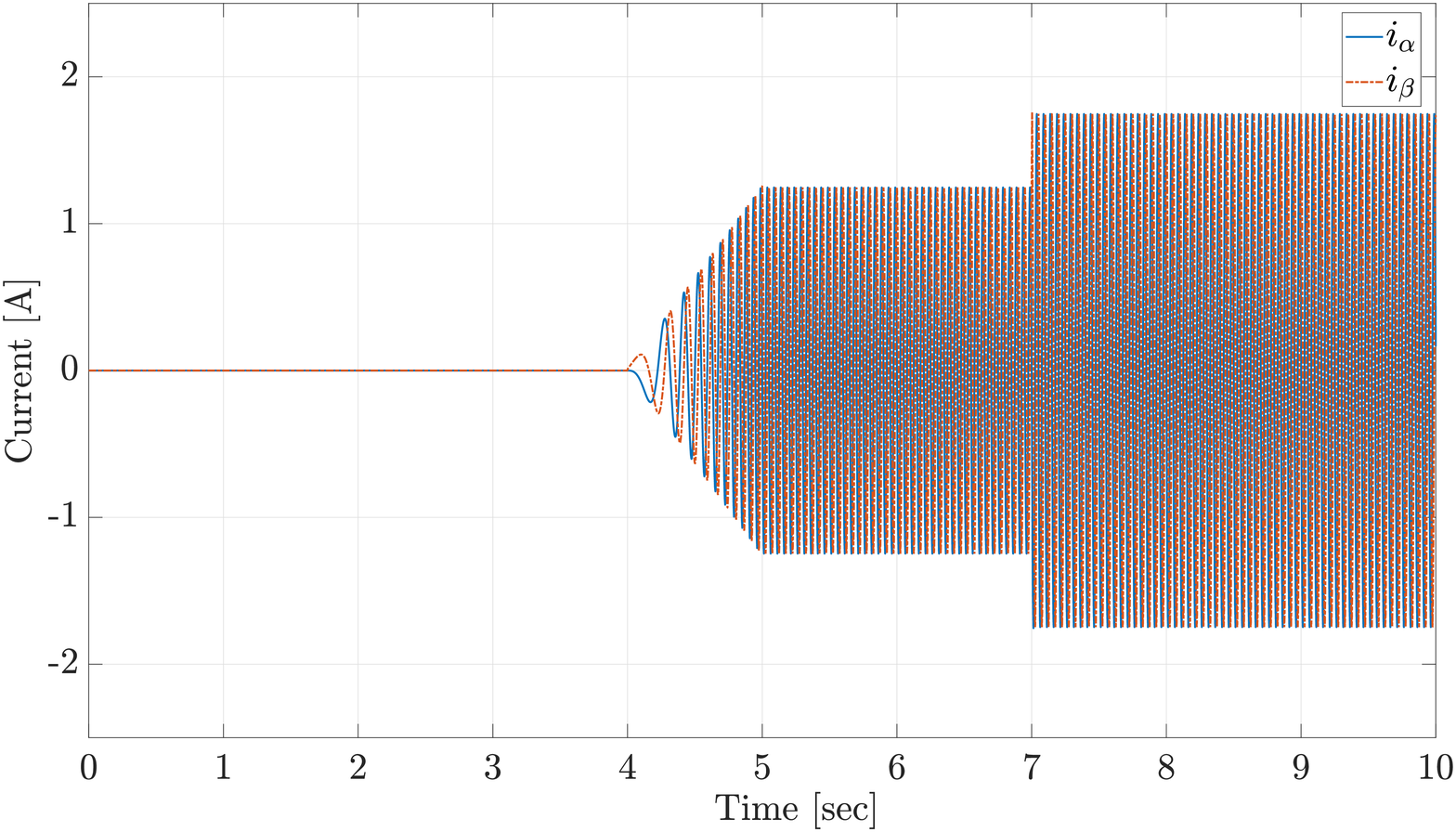}   % \includegraphics[width=0.7\columnwidth, height=2.8in]{Psi2.jpg}
\caption{Currents of the LPV controller in the $\alpha - \beta$ axis.}
\label{fig:distcurrent}
\end{figure}
\begin{figure}[t]
\centering
\includegraphics[width=\columnwidth,height=2.3in]{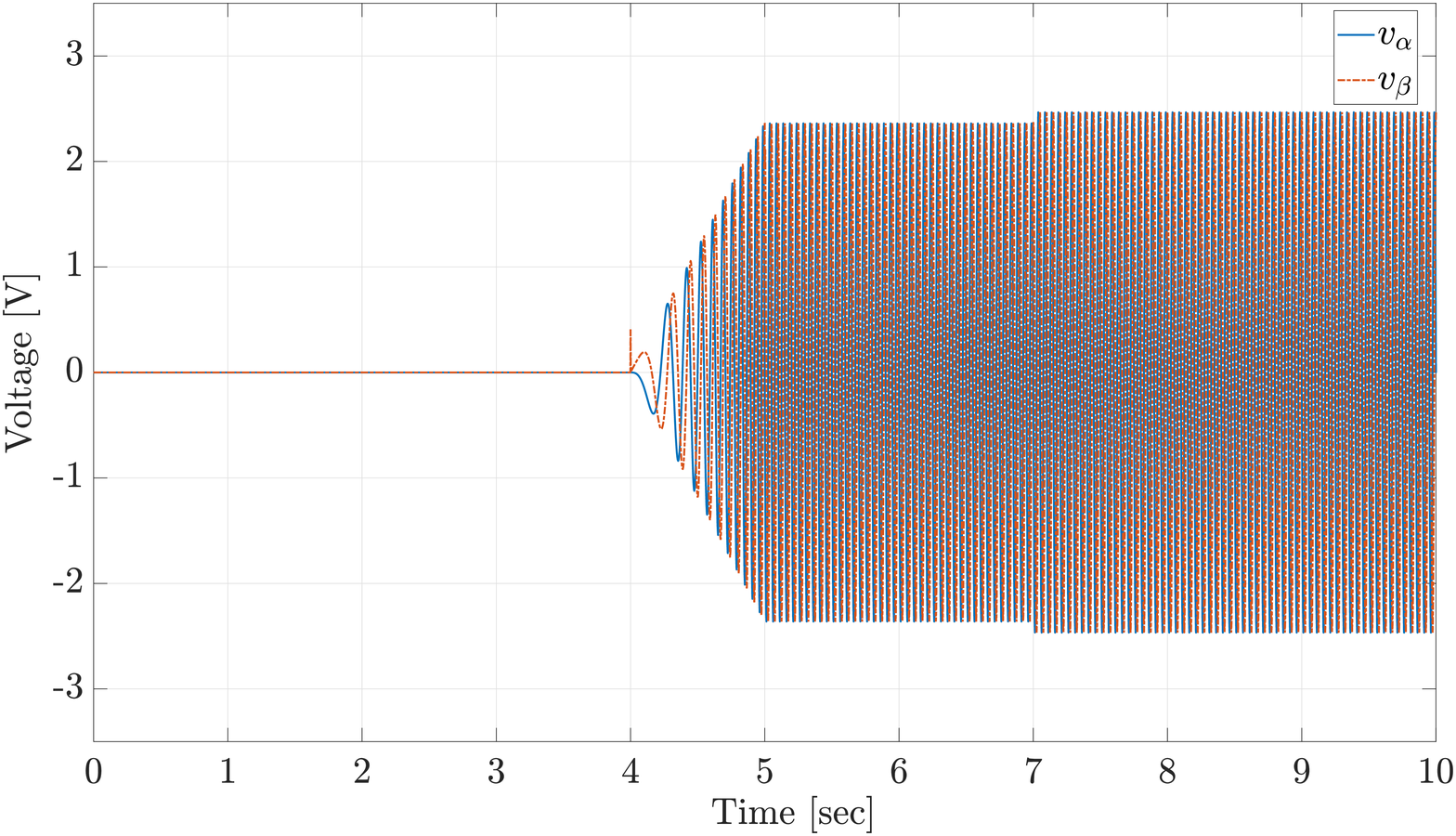}   % \includegraphics[width=0.7\columnwidth, height=2.8in]{Psi2.jpg}
\caption{Control input voltages of the LPV controller in the $\alpha - \beta$ axis.}
\label{fig:distvoltage}
\end{figure}
\begin{figure}[t]
\centering
\includegraphics[width=\columnwidth, height=2.4in]{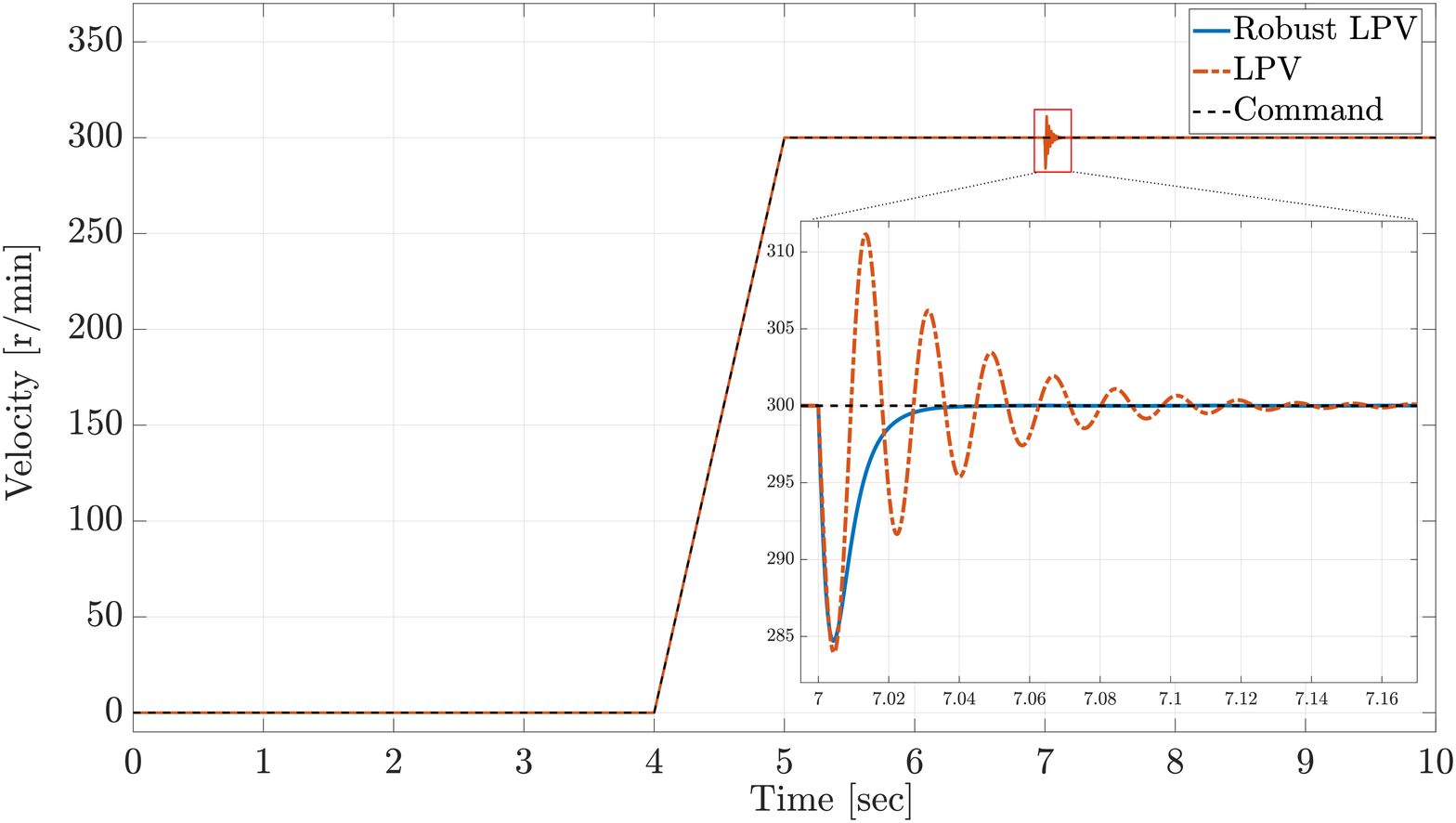}   % \includegraphics[width=0.7\columnwidth, height=2.8in]{Psi2.jpg}
\caption{Closed-loop velocity tracking performance of the robust LPV controller in the presence of model parameter uncertainty subject to torque load disturbance.}
\label{fig:trackingdist_robust}
\end{figure}

%\bibliographystyle{SageH.bst}
%\bibliography{ref}
%\end{document}

\end{document}